\documentclass[runningheads]{llncs}
\usepackage{verbatim}
\usepackage[T2A]{fontenc}
\usepackage{amsmath, graphicx}
\usepackage[utf8]{inputenc} 
\usepackage { tikz } 
\usetikzlibrary{arrows.meta} 
\usetikzlibrary{decorations.pathreplacing}

\let\no\textnumero
\def\rank{\mathop{rk}}
\title{Communication-efficient parallel Bruhat decomposition}
\author{Ioanna Evtushevskaya-Konovalova  \and
Alexander Tiskin }
\institute{St Petersburg State University}

\begin{document}
\maketitle

\begin{abstract}
The model of bulk-synchronous parallel (BSP) computation is an emerging paradigm of general-purpose parallel computing. Bruhat decomposition is an important method in numerical linear algebra, generalising ordinary LU decomposition while providing a symmetric structured way of expessing pivoting. Block-recursive algorithms have been developed in the past for various numerical linear algebra problems, including Bruhat decomposition; however, expressing recursive algorithms in the BSP model still remains a challenge. In this paper, we consider the communication and synchronisation complexity of Bruhat decomposition in the BSP model. We develop and analyze two communication-efficient algorithms based on different recursion schemes.
\end{abstract}

\section{Introduction}

Bruhat decomposition is a representation of a non-singular square matrix \(A\) as a product \(A = U_1 P U_2\), where \(U_1, U_2\) are upper-triangular matrices and \(P\) is a permutation matrix. Matrix \(P\) is uniquely determined. Bruhat decomposition is a natural generalization of the \(LU\) decomposition: indeed, if all leading principal minors of matrix $A$ are non-singular, the matrix $P$  becomes anti-diagonal, and then the product \(U_1P\) is precisely the lower-triangular matrix \(L\) of the decomposition $A=LU$. Classical algorithms for parallel computation of the \(LU\) decomposition are described e.g.\ in \cite{Higham:02}. In this work, we will construct a more general algorithms for Bruhat decomposition. 

For convenience, we reformulate the problem slightly. Let a non-singular square matrix \(A\) be given. We want to find non-singular matrices \(U\) (upper-triangular), \(L\) (lower-triangular), and \(P\) (permutation) such that the following holds:
\[
LAU = P.
\]
Indeed, given these matrices, it is easy to compute matrices for  Bruhat decomposition in its classical form. 

In this work, we develop and analyze two algorithms based on different recursion schemes. We describe them from the BSP model perspective, including parallelization and cost analysis. 

Throughout the paper, we ignore small irregularities arising from imperfect matching of parameters. For example, when we say that $p$ processors are partitioned into two equal groups, we assume implicitly that the groups' sizes may differ by $\pm 1$.

\section{Existing work}
Malaschonok in \cite{Malaschonok:10} describes parallel Bruhat decomposition algorithms that rely on the generalized Bruhat decomposition. His algorithm was considered outside any particular parallel computing model. We will consider a similar algorithm in the BSP model and we will estimate this algorithm's costs. Our second algorithm relies heavily on banded matrices, and the connection of the Bruhat decomposition to such matrices is described by Strang \cite{Strang:15}. However, Strang does not discuss his algorithm from the standpoint of parallel computation.
Generalized Bruhat decomposition was introduced by Grigoriev in \cite{Grigoriev:82_TCS} and Malaschonok in \cite{Malaschonok:10}. 
This decomposition preserves, in a certain sense, the rank stucture of matrix \(A\). 
\begin{definition}
     A \emph{subpermutation matrix} is a matrix, where every row and every column contains at most one non-zero entry, and every non-zero entry is equal to 1.
\end{definition}
\begin{definition}
 Given an $n \times n$ matrix \(A\), let  \(i_1, \dots, i_m\), \(j_1, \dots, j_k\) be the indices of its zero rows and columns. An $LEU$-decomposition is \(A = L E U\), where \(U\), \(L\) are non-singular \(n \times n\) upper- (respectively lower-)triangular matrices, and \(E\) is an \(n \times n\) subpermutation matrix, \(i_1, \dots, i_m\), \(j_1, \dots, j_k\) are the indices of zero rows and columns of matrix $E$, and \(\rank(E) = \rank(A)\).  Columns \(i_1, \dots, i_m\) of \(L\) and rows \(j_1, \dots, j_k\) of \(U\) coincide with the columns and rows of the identity matrix.
\end{definition}
Every square matrix has an \(LEU\)-decomposition. (This is proved e.g.\ in \cite{Malaschonok:10}.)

Various generalizations of the Bruhat decomposition have also been researched by Grigoriev  \cite{Grigoriev:82_TCS}, Dumas et al.  \cite{Dumas:17}. Odeh et al.\  \cite{Odeh:98}. 
Tiskin \cite{Tiskin:07_FGCS} has proposed BSP algorithms for the related problems of Gaussian elimination with pairwise pivoting and QR decomposition.

\section{Bulk Synchronous Parallelism}
\subsection{The BSP model}
We consider parallel computations in the BSP (Bulk Synchronous Parallel) model. This model was proposed by Valiant \cite{Valiant:90_CACM}; it is also described e.g.\ by Tiskin \cite{Tiskin:98_TCS}. The BSP model assumes several processors that can perform computations in parallel and exchange information with each other. Each processor has a local memory, and there is also a global memory accessible to all processors. We will denote the number of processors by $p$. For information exchange, a network is assumed that routes communication between pairs of processors. Parameter \(g\) is introduced as the reciprocal of network bandwidth; that is, sending/receiving a unit of information takes \(g\) conventional time units.
Thus, in this model, it is important to optimize not only the distribution of computation blocks but also the communication between them.

The computation in the BSP model is a sequence of supersteps. At the end of each superstep, a barrier synchronization occurs, ensuring that all processors have completed the current stage of local computation and network communication, and are ready to proceed to the next stage of computation. Thus, during a superstep, a processor cannot use the results of computations by other processors obtained in the same superstep, and is guaranteed to receive them only after the synchronization at the end of that superstep. Synchronization also signals that the superstep is complete and the next one can begin. Another parameter \(l\) is introduced as the cost of synchronization.
Thus, the number of supersteps (each involving a synchronization) also needs to be optimized.

Thus, the system has three main parameters: \(p\), \(g\),  \(l\).

We will use the following notation:
\begin{itemize}
    \item \(W(k,r)\) – computation cost on processor \(r\) at superstep \(k\),
    \item \(H(k,r)\) – communication cost on processor \(r\) at superstep \(k\),
    \item \(S\) – number of supersteps.
\end{itemize}

The computation and communication costs at superstep \(k\) are:
\[
W(k) = \max_r (W(k,r))\qquad
H(k) = \max_r (H(k,r))
\]

Thus, the total cost is calculated as:
\[
\sum_{i=0}^{S} \left( W(i) + g \cdot H(i) + l \right)
\]
In the BSP model, the computation cost must be evenly balanced. Specifically, if a sequential algorithm for a computational problem has complexity $T(n)$, then the corresponding parallel algorithm must achieve a computational cost of  $T(n)/p$ . Subject to this condition, the communication and synchronization costs are to be minimized.
It is also standard to assume the slackness condition $n \gg p$, which we adopt in this paper.

\subsection{BSP matrix multiplication}
When computing the Bruhat decomposition and the LEU decomposition, it is necessary to compute matrix products in parallel. Matrix multiplication can be carried out in parallel as follows, it can also be computed in parallel using the algorithm developed in
\cite{McColl-Tiskin:99}.  Let $A, B$ be $n \times n$ matrices, $C = AB$, and $p$ be the number of processors. We partition the matrices into $p^{2/3}$ square blocks of size $n / p^{1/3} \times n / p^{1/3}$. Denote the block size as $d = n / p^{1/3}$. The block of matrix $A$ formed by the intersection of rows $d(i-1)+1, d(i-1)+2, \dots, di$ and columns $d(j-1)+1, d(j-1)+2, \dots, dj$ is denoted by $A[i, j]$, and similarly for matrices $B$ and $C$. Then it is clear that

\[
C[i, j] = \sum_{k=1}^{p^{1/3}} A[i, k] B[k, j].
\]

Thus, it is necessary to compute $A[i, k] B[k, j]$ for all $1 \leq i, k, j \leq p^{1/3}$. Each processor is assigned a triple $(i, k, j)$ and computes the product $C[i, k, j] = A[i, k] B[k, j]$, writing the result to external memory. Matrix $C$ is obtained by summing the corresponding blocks: $C[i, j] = \sum_{k=1}^{p^{1/3}} C[i, k, j]$.

Each processor performs $O\bigl((n / p^{1/3})^3\bigr) = O(n^3 / p)$ elementary multiplications when computing the product of two blocks of size $n / p^{1/3}$, and sends one computed square block of the same size i.e.\ $O\bigl((n / p^{1/3})^2\bigr) = O(n^2 / p^{2/3})$. Thus, the costs for computation, communication, and synchronization are:

\[
\begin{aligned}
W(n, p) &= O(n^3 / p), \\
H(n, p) &= O(n^2 / p^{2/3}), \\
S(n, p) &= O(1).
\end{aligned}
\]

\section{Block-recursive algorithm}

We will solve the problem in the following form: it is required to find matrices \(L\), \(U\) such that
\[
LAU = P.
\]
Then, by computing the inverses \(L^{-1}\), \(U^{-1}\), we obtain the Bruhat decomposition. We will solve the problem iteratively. At iteration $k$ of the algorithm, we eliminate some elements of the current matrix $A^{(k)}$ by multiplication on the left with lower-triangular \(L^{(k)}\) and on the right with upper-triangular \(U^{(k)}\). Thus, we obtain a matrix \(A^{(k+1)}\) that will gradually become a permutation matrix. Matrix \(A^{(k)}\) is of the form
\[
\begin{pmatrix}
e^{(k)} & a_{01}^{(k)} \\
a_{10}^{(k)} & a_{11}^{(k)} 
\end{pmatrix}
\]
Here \(e^{(k)}\) is a block of size \(m \times m\) (where \(m \) is to be specified later.) This block is a subpermutation matrix. It may contain entirely zero rows or columns; we will call them \emph{special} and handle them in a special way. Note that the blocks \(a_{01}^{(k)}\) and \(a_{10}^{(k)}\) are rectangular, in general not square. In the following description of an individual algorithm iteration, we will reason about square blocks of size \(m\) and \(2m\). For convenience, blocks of size \(m\) are denoted by lowercase letters, and blocks of size \(2m\) by uppercase letters. Each block is processed  in parallel using recursive calls on smaller blocks, down to a certain sufficiently small block size \(n_0 \). Blocks of this size will be processed sequentially. 

\subsection{Recursive scheme}
The main idea is that when processing a certain block, we can reduce it to several matrix multiplications and the computation of an \(LEU\)-decomposition on smaller blocks. Computation will be distributed among processors, and at the bottom of the recursion tree, each processor will compute the $LEU$-decomposition for a separate block that is sufficiently small so that the $LEU$-decomposition can be computed sequentially. 

Our algorithm is described as a sequence of iterations. From the point of view of parallel computations, an iteration looks as follows: first, the $LEU$-decomposition for the upper-left block of the matrix is computed recursively; then the processors receive the bottom-left and upper-right blocks and compute certain auxiliary matrices. Each of these matrices is computed using the BSP matrix multiplication algorithm from Subsection 3.2; moreover, different matrices can also be computed in parallel by different groups of processors. Then the bottom-left and upper-right blocks are updated. Next, the \(LEU\)-decomposition is computed recursively for the bottom-left and upper-right blocks. After computing the \(LEU\)-decomposition for the bottom-left and upper-right blocks and updating the lower-right block via matrix multiplications, we make one more recursive call and compute the \(LEU\)-decomposition for this block. Finally, we obtain the $LEU$-decomposition for the entire matrix from block $LEU$-decompositions, again using the BSP matrix multiplications.

The algorithm is illustrated in Figure~\ref{f-block}.

\begin{figure}[tb]
\centering
\begin{tikzpicture}[scale=0.6] 
\filldraw[color=black!60, fill=gray!50, very thick] (0, 3) rectangle (4, -1);
\filldraw [color=black!60, fill=green!50, very thick] (1,2) rectangle (0,3);
\filldraw [color=black!60, fill=gray!50, very thick] (2,1) rectangle (1,2);
\filldraw [color=black!60, fill=gray!50, very thick] (1, 2) rectangle (2,3);
\filldraw [color=black!60, fill=gray!50, very thick] (0,1) rectangle (1,2);
\end {tikzpicture} 
\qquad
\begin {tikzpicture}[scale=0.6] 
\filldraw[color=black!60, fill=gray!50, very thick] (0, 3) rectangle (4, -1);
\filldraw [color=black!60, fill=green!50, very thick] (1,2) rectangle (0,3);
\filldraw [color=black!60, fill=gray!50, very thick] (2,1) rectangle (1,2);
\filldraw [color=black!60, fill=green!50, very thick] (1, 2) rectangle (2,3);
\filldraw [color=black!60, fill=green!50, very thick] (0,1) rectangle (1,2);
\end {tikzpicture} 
\qquad
\begin {tikzpicture}[scale=0.6] 
\filldraw[color=black!60, fill=gray!50, very thick] (0, 3) rectangle (4, -1);
\filldraw [color=black!60, fill=green!50, very thick] (1,2) rectangle (0,3);
\filldraw [color=black!60, fill=green!50, very thick] (2,1) rectangle (1,2);
\filldraw [color=black!60, fill=green!50, very thick] (1, 2) rectangle (2,3);
\filldraw [color=black!60, fill=green!50, very thick] (0,1) rectangle (1,2);
\end {tikzpicture} 
\qquad
\begin {tikzpicture}[scale=0.6] 
\filldraw[color=black!60, fill=gray!50, very thick] (0, 3) rectangle (4, -1);
\filldraw [color=black!60, fill=green!50, very thick] (1,2) rectangle (0,3);
\filldraw [color=black!60, fill=green!50, very thick] (2,1) rectangle (1,2);
\filldraw [color=black!60, fill=green!50, very thick] (1, 2) rectangle (2,3);
\filldraw [color=black!60, fill=green!50, very thick] (0,1) rectangle (1,2);
\filldraw [color=black!60, fill=green!50, very thick] (0,-1) rectangle (2, 1);
\filldraw [color=black!60, fill=green!50, very thick] (4, 3) rectangle (2,1);
\filldraw [color=black!60, fill=green!50, very thick] (0, 3) rectangle (2,1);
\end {tikzpicture}
\caption{\label{f-block} Block-recursive algorithm}
\end{figure}
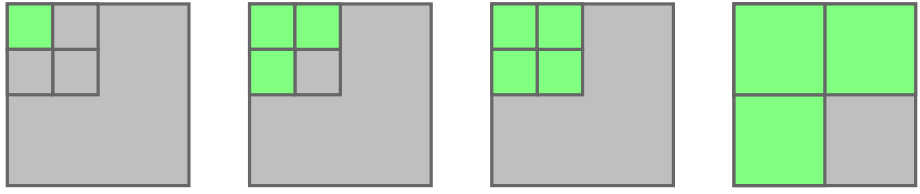

Consider the process of reducing a matrix to  subpermutation form using \(LEU\)-decomposition.
Suppose that at some point we have obtained a matrix \(B\) of size \(m \times m\) partitioned into blocks of size \(m/2 \times m/2\):
\[
B = \begin{pmatrix} b_{00} & b_{01} \\ b_{10} & b_{11} \end{pmatrix}
\]

\paragraph{1. \(LEU\)-decomposition for \(b_{00}\)}

This is obtained recursively as
\[
l_{00} b_{00} u_{00} = e_{00}
\]
where \(e_{00}\) is a subpermutation matrix. Transformation matrices are formed as
\[
L_1 = \begin{pmatrix} l_{00} & 0 \\ 0 & I \end{pmatrix} \qquad
U_1 = \begin{pmatrix} u_{00} & 0 \\ 0 & I \end{pmatrix}
\]
where \(I\) is the identity matrix of conforming size.

After applying these transformations, the current matrix takes the form
\[
L_1BU_1=
B_1 = \begin{pmatrix} e_{00} & b_{01}' \\ b_{10}' & b_{11} \end{pmatrix}
\]
\paragraph{2. Partial elimination of rows in \(b_{01}'\) and columns in \(b_{10}'\)}

After the previous step, we have some non-zero rows/columns in $e_{00}$.  If $e_{00}$ has a non-zero column indexed $i$, we can obtain a zero column indexed $i$ in $b_{10}'$. Similarly, we can obtain a zero row indexed $i$ in $b_{01}'$. 
Thus, elements in blocks \(b_{01}'\) and \(b_{10}'\) are eliminated using lower-triangular and upper-triangular matrices:
\[
L' = \begin{pmatrix} I & 0 \\ l' & I \end{pmatrix} \qquad
U' = \begin{pmatrix} I & u' \\ 0 & I \end{pmatrix}
\]
where
\[
l' = -b_{10}' e_{00}^{\top}, \qquad u' = -b_{01}' e_{00}^{\top}.
\]
After applying these transformations, the current matrix takes the form
\[
L'B_1U'=B_2 = \begin{pmatrix} e_{00} & b_{01}^{*} \\ b_{10}^{*} & b_{11}^{'} \end{pmatrix}
\]

\paragraph{3. \(LEU\)-decomposition for \(b_{01}^{*}\) and \(b_{10}^{*}\)}

This is obtained recursively as
\[
l_{01}^{*} b_{01}^{*} u_{01}^{*} = e_{01}^{*}\qquad
l_{10}^{*} b_{10}^{*} u_{10}^{*} = e_{10}^{*}
\]
Transformation matrices are formed as
\[
L_2 = \begin{pmatrix} l_{01}^{*} & 0 \\ 0 & l_{10}^{*} \end{pmatrix} \qquad
U_2 = \begin{pmatrix} u_{01}^{*} & 0 \\ 0 & u_{10}^{*} \end{pmatrix}
\]
After applying these transformations, the current matrix takes the form
\[
L_2B_2U_2 = B_3 = \begin{pmatrix} e_{00} & e_{01} \\ e_{10} & b_{11}'' \end{pmatrix}
\]

\paragraph{4. Partial elimination in rows and columns of block \(b_{11}''\)}

Similarly to step 2, we obtain matrices \(L'', U''\) 
\[
L'' = \begin{pmatrix} I & 0 \\ l'' & I \end{pmatrix} \qquad
U'' = \begin{pmatrix} I & u'' \\ 0 & I \end{pmatrix}
\]
\[
l'' = -b_{11}'' p_{01}^{\top}, \qquad u'' = -b_{11}''p_{10}^{\top}
\]

After applying these transformations, the current matrix takes the form
\[
L''B_3U''=B_3' = \begin{pmatrix} e_{00} & e_{01} \\ e_{10} & b_{11}^* \end{pmatrix} 
\]
\paragraph{5. \(LEU\)-decomposition for the \(b_{11}^*\)} 

This is obtained recursively as
\[
l_{11}^{*} b_{11}^{*} u_{11}^{*} = e_{11}^{*}
\]
with transformation matrices
\[
L_3 = \begin{pmatrix} I & 0 \\ 0 & l_{11}^{*} \end{pmatrix} \qquad
U_3 = \begin{pmatrix} I & 0 \\ 0 & u_{11}^{*} \end{pmatrix}
\]
After applying these transformations, the current matrix takes the form
\[
L_3B_3U_3=B_4 = \begin{pmatrix} e_{00} & e_{01} \\ e_{10} & e_{11} \end{pmatrix} = E
\]

\paragraph{Summary}
Overall, the original matrix \(B\) is reduced to sub-permutation form by successively multiplying on the left by lower-triangular matrices and on the right by upper-triangular matrices:
$
L_1 L' L_2 L'' L_3 \; B \; U_3 U'' U_2 U' U_1 = E,
$
where \(E\) is the resulting sub-permutation matrix. This is equivalent to finding the \(LEU\)-decomposition of \(B\). If $B$ is non-singular, than $E$ is a permutation matrix.

\subsection{Cost analysis}

We can derive the following recurrence relations for our algorithm's costs of computation, communication, and synchronization. Let $m$ be the matrix size and $q$ the number of processors in the current recursive call. 
As mentioned earlier, at the bottom of the recursion tree, each processor will process a separate block that is sufficiently small to be processed sequentially. The size of this block can be varied. On one hand, it must be small enough for uniform division of computation. 
On the other hand, partitioning the matrix into overly small blocks increases synchronization cost. 

Let $n_0 = n / p^{\alpha}$ be the block size, where $\alpha$ is a parameter to be defined later. By varying $\alpha$, we select the optimal block size $n_0$. Obviously, the recursion depth depends linearly on $\alpha$. The computation, communication and synchronization costs in this scheme can be described by the following recurrences:
\begin{align*}
W(m, q) &= 2W(m/2, q) + W(m/2, q/2) + O(m^3/q) \\
H(m, q) &= 2H(m/2, q) + H(m/2, q/2) + O(m^2/q^{2/3})\\
S(m, q) &= 2S(m/2, q) + S(m/2, q/2) + O(1)
\end{align*}
Indeed, processing a block of size $m$ with $q$ processors proceeds as follows: first, all $q$ processors process the upper-left block of size $m/2$; then the processors are split into two equal groups of $q/2$ processors and process the upper-right and the bottom-left blocks of size $m/2$ in parallel; then all $q$ processors process the bottom-right block of size $m/2$. Moreover, in addition to recursive calls, matrix multiplications are needed to update the blocks. 
The base values for the recurrences are provided by the respective costs of sequential computation:
\begin{align*}
W(n_0, q) = O(n_0^3) \qquad
H(n_0, q) = O(n_0^2) \qquad
S(n_0, q) = O(1)
\end{align*}

From the above recurrences, we can obtain explicit expressions for $W(n, p)$, $H(n, p)$ and $S(n, p)$.

When analyzing algorithms, we will always omit the base of the logarithm, implying the binary logarithm.
Let $\mu = \log(m)$, $\lambda = \log(q)$ and introduce new functions:
\[
\tilde{W}(\mu, \lambda) = W(m, q) \qquad  
\tilde{H}(\mu, \lambda)  = H(m, q) \qquad
\tilde{S}(\mu, \lambda)  = S(m, q)
\]
Let $k = \log(n)$. Then the recurrences become:
\begin{gather*}
\tilde{W}(\mu, \lambda) = 2\tilde{W}(\mu - 1, \lambda) + \tilde{W}(\mu - 1, \lambda - 1) + O\left(\frac{2^{3\mu}}{2^\lambda}\right)\\
\tilde{H}(\mu, \lambda) = 2\tilde{H}(\mu - 1, \lambda) + \tilde{H}(\mu - 1, \lambda - 1) + O\left(\frac{2^{2\mu}}{2^{2\lambda/3}}\right)\\
\tilde{S}(\mu, \lambda) = 2\tilde{S}(\mu - 1, \lambda) + \tilde{S}(\mu - 1, \lambda - 1) + O(1)
\end{gather*}

These recurrences gives rise to a dag (directed acyclic graph) whose vertices are points in the 2D integer lattice. Their coordinates are interpreted  as follows.  The horizontal coordinate of a vertex equals the logarithm of the number of processors; the vertical coordinate equals the logarithm of the size of the block being processed. Each recursion level consists of vertices with a fixed vertical coordinate. Deepening the recursion corresponds to moving from the root up and to the left, as shown in Figure~\ref{f-tree}. 
Each vertex of the dag has two child vertices. 
In particular, vertex $(\mu, \lambda)$ has two child vertices: 
$(\mu - 1, \lambda)$, connected to the parent by a vertical edge, 
and $(\mu - 1, \lambda - 1)$, connected to the parent by a diagonal edge. 
All the vertical edges are doubled, while diagonal ones are not. (This corresponds to coefficients 1 and 2 in the recurrence.)

\begin{figure}[tb]
    \centering
    
\begin{tikzpicture}[scale=0.6]
  \node[circle, fill=white, scale=0.5pt] (x) at (2,-1) {}; 

  \node[circle, fill= white, scale=0.5pt] (y) at (11,6) {};
  \node[circle, fill= black, scale=0.5pt] (lambda0) at (10,6) {};
  \node[circle, fill= white] (Q) at (10,6.6) {$\lambda$};
  
  \node[circle, fill=white] (S) at (6,6.6) {$\lambda-2$};
  \node[circle, fill=white] (a) at (8,6.6) {$\lambda-1$};
  \node[circle, fill=black, scale=0.5pt] (lambda2) at (6,6){};
  \node[circle, fill=black, scale=0.5pt] (lambda1) at (8,6){};
  \node[circle, fill=black, scale=0.5pt] (t) at (2,2) {};
  \node[circle, fill=black, scale=0.5pt] (k) at (2,4) {};
  \node[circle, fill=white, scale=0.5pt] (z) at (2,6) {};

  \node[circle, fill=white] (t) at (1,2) {$\mu-1$};
  \node[circle, fill=white] (k) at (1,4) {$\mu-2$};
  \node[circle, fill=white] (i) at (1,0) {$\mu$};
  \node[circle, fill=black, scale=0.5pt] (i) at (2,0) {};
  
    \node[circle, fill=red, scale=0.5pt] (A) at (10,0) {};
    
     \node[circle, fill=red, scale=0.5pt] (B) at (8,2) {};
     \node[circle, fill=red, scale=0.5pt ] (C) at (10,2) {};

     \node[circle, fill=red, scale=0.5pt ] (D) at (10,4) {};

     \node[circle, fill=red, scale=0.5pt] (E) at (8,4) {};

     \node[circle, fill=red, scale=0.5pt] (F) at (6,4) {};

    \draw[->, thick] (A) -- (B);
    \draw[->, double] (A) -- (C);
    \draw[->, thick] (B) -- (F);
    \draw[->, double] (B) -- (E);
    \draw[->, double] (C) -- (D);
    \draw[->, thick] (C) -- (E);
    \draw[->, thick] (z) -- (y);
    \draw[->, thick] (z) -- (x);
\end{tikzpicture}
    \caption{\label{f-tree} Recursion dag for the block-recursive algorithm}
\end{figure}
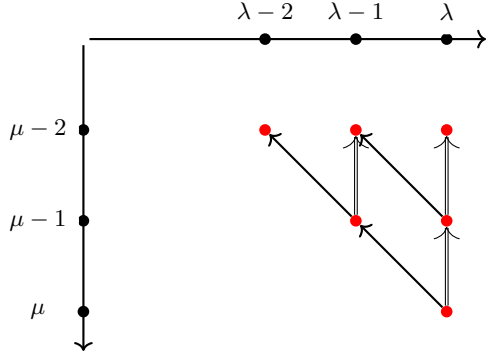

To express the costs $\tilde{W}(\mu, \lambda)$, $\tilde{H}(\mu, \lambda)$, $\tilde{S}(\mu, \lambda)$ we write for a vertex $v$ with coordinates $(\mu, \lambda)$
\begin{gather*}
w(v) = O\left(\frac{2^{3\mu}}{2^\lambda}\right) \qquad
h(v) = O\left(\frac{2^{2\mu}}{2^{2\lambda/3}}\right) \qquad
s(v) = O(1) \\
r(v) = \text{the number of paths from the root to $v$}
\end{gather*}
We have
\begin{gather*}
\textstyle
\tilde{W}(\nu, \pi) = \sum_v w(v) r(v) \qquad
\tilde{H}(\nu, \pi) = \sum_v h(v) r(v) \qquad
\tilde{S}(\nu, \pi) = \sum_v s(v) r(v)
\end{gather*}
where the sums range across all the vertices of the recursion dag.
Consider level $i$ of this dag: vertices at this level have coordinates of the form $(k-i, i-j)$ for all $0 \leq j \leq i$. A path to such a vertex is of length $i$ with exactly $j$ vertical edges. Since vertical edges are doubled, traversing a vertical edge doubles the number of paths. Thus, the total number of paths equals $\binom{i}{j} 2^j$.

Next, we will separately sum up the terms corresponding to vertices of each level. First, we obtain the computation cost. Summing up the terms $w(v)$ corresponding to vertices of level $i$:
\begin{gather*}
\textstyle
\tilde{W}_i = \sum_{v=(k-i,j), j=0,\dots,i} w(v) r(v) =
\sum_{j=0}^{i} \binom{i}{j} 2^{j+3(k-i)+(i-j)} \frac{1}{p} =\\ 
\textstyle
\frac{1}{p} \sum_{j=0}^{i} \binom{i}{j} 2^{3k-2i} =
2^{3k-2i} \frac{1}{p} \sum_{j=0}^{i} \binom{i}{j} = 2^{3k-i} \frac{1}{p}
\end{gather*}
Then, summing up over all levels, we obtain:
\[
\textstyle\sum_{i=0}^{k} \tilde{W}_i = 2^{3k} \sum_{i=0}^{k} 2^{-i} \frac{1}{p} = O(2^{3k}) \frac{1}{p} = O\left(\frac{n^3}{p}\right)
\]

Now we similarly obtain the communication cost. Summing the terms corresponding to vertices of level $i$, we obtain
\begin{gather*}
\tilde{H}_i = \textstyle \sum_{v=(k-i,j), j=0,\dots,i} h(v) r(v) = \sum_{j=0}^{i} \binom{i}{j} 2^{j+2(k-i)+2(i-j)/3} \cdot \frac{1}{p^{2/3}} \\  \textstyle= \sum_{j=0}^{i} \binom{i}{j} 2^{2k-4i/3 + j/3} \cdot \frac{1}{p^{2/3}}  = 2^{2k-4i/3} \frac{1}{p^{2/3}} \sum_{j=0}^{i} \binom{i}{j} 2^{j/3} = 2^{2k-4i/3} \frac{1}{p^{2/3}} (1 + 2^{1/3})^i
\end{gather*}
Then, summing over all levels, we obtain:
\begin{gather*}
\textstyle \sum_{i=0}^{k} \tilde{H}_i = 2^{2k} \frac{1}{p^{2/3}} \sum_{i=0}^{k} 2^{-4i/3} (1+2^{1/3})^i = 2^{2k} \frac{1}{p^{2/3}} \sum_{i=0}^{k} \bigl(2^{-4/3} (1+2^{1/3})\bigr)^i \approx\\\textstyle\approx 2^{2k} \frac{1}{p^{2/3}} \sum_{i=0}^{\infty} 0.89^i \approx 2^{2k+3} \frac{1}{p^{2/3}} = O\left(\frac{n^2}{p^{2/3}}\right)
\end{gather*}

Now we similarly obtain the synchronization cost:
\[
S(n, p) =\textstyle O\Bigl(\left(\sum_{i=0}^{k} \left(\sum_{j=0}^{i} 2^j \binom{i}{j}\right)\right)\Bigr) = O\Bigl( \sum_{i=0}^{k} 3^i\Bigr) = O(3^k)
\]

Overall we have obtained the following asymptotic costs:
\begin{gather*}
\textstyle W(n, p) = O\left(\frac{n^3}{p^{2/3}}\right)\\
\textstyle H(n, p) = O\left(\frac{n^2}{p^{2/3}}\right) + O(3^k) = O\left(\frac{n^2}{p^{2/3}}\right)\\\textstyle
S(n, p) = O(3^k) = O(n^{\log3}) \approx O(n^{1.58})
\end{gather*}

Recall that the recursion tree is evaluated in parallel while the block size is large,
and switches to sequential evaluation once a critical block size $n_0 = n/p^\alpha$ is reached.
Denote $k_0 = \log(n/n_0)$. For every vertex \(u=(k_0, \lambda - j)\), we delete all its descendants and make it a leaf. We have
\[
\begin{aligned}
w(u) = O(2^{3k_0}) \qquad
h(u) = O(2^{2k_0}) \qquad
s(u) = O(1) \qquad
r(u) = r(u)
\end{aligned}
\]

We denote the set of inner vertices (vertices at levels $0$ to $k_0-1$) and leaves (level $k_0$) by $\mathit{Inner}$, $\mathit{Leaf}$, respectively.

Using the previous analysis of the recurrence dag, we can bound the contribution of the $\mathit{Inner}$ terms as follows:
\begin{gather*}
\textstyle\tilde{W}_{\mathit{inner}} \leq O\left(\frac{n^3}{p}\right)\qquad
\textstyle\tilde{H}_{\mathit{inner}} \leq O\left(\frac{n^2}{p^{2/3}}\right)
\end{gather*}
The bound on synchronization is obtained from the previous analysis by substituting $k_0$ for $k$:
\[
\tilde{S}_{\mathit{inner}} = \textstyle O\left(\sum_{i=0}^{k_0} \left(\sum_{j=0}^{i} 2^j \binom{i}{j}\right)\right) = O( \sum_{i=0}^{k_0} 3^i) = O(3^{k_0}) = O\left(\left(\frac{n}{n_0}\right)^{\log 3}\right)
\] 

Independently, we sum up the $\mathit{Leaf}$ terms:
\begin{align*}
\tilde{W}_{\mathit{leaf}} &\textstyle= \sum_{v \in \mathit{Leaf}} w(v) r(v) = \sum_{j=0}^{k_0} \binom{k_0}{j} 2^{3(k-k_0)} 2^j \\
&\textstyle= 2^{3(k-k_0)} \sum_{j=0}^{k_0} \binom{k_0}{j} 2^j = 2^{3(k-k_0)} 3^{k_0}
\\
&\textstyle= 2^{3(\log n - \log(n/n_0))} 3^{\log(n/n_0)} = n_0^3 \left(\frac{n}{n_0}\right)^{\log 3} \\
&\textstyle= O\left(\left(\frac{n}{p^\alpha}\right)^3\right) p^{\alpha \log 3} = O\left(\frac{n^3}{p^{\alpha(3-\log 3)}}\right)\\
\tilde{H}_{\mathit{leaf}} &\textstyle= \sum_{v \in \mathit{Leaf}} h(v) r(v) = \sum_{j=0}^{k_0} \binom{k_0}{j} 2^{j+2(k-k_0)} = 2^{2(k-k_0)} \sum_{j=0}^{k_0} \binom{k_0}{j} 2^j \\ 
&\textstyle= O(n_0^2) 3^{k_0} = O\left(\left(\frac{n}{p^\alpha}\right)^2\right) p^{\alpha \log 3}\\
\tilde{S}_{\mathit{leaf}} & = O(3^{k_0})
\end{align*}
Since we need $\tilde{W}_{\mathit{leaf}} = O\left(\frac{n^3}{p}\right)$, it must hold $\alpha \geq \alpha_{\min} = 1/(3-\log 3) \approx 0.70$.

Finally, we sum up the $\mathit{Inner}$ and $\mathit{Leaf}$ expressions together:
\begin{align*}
&\textstyle W(n, p) = \tilde{W}_{\mathit{inner}} + \tilde{W}_{\mathit{leaf}} = O\left(\frac{n^3}{p}\right) \\ 
& \textstyle H(n, p) = \tilde{H}_{\mathit{inner}} + \tilde{H}_{\mathit{leaf}} =\\
&\qquad= O\left(\frac{n^2}{p^{2/3}}\right) + O\left(\left(\frac{n}{p^\alpha}\right)^2\right) p^{\alpha \log(3)} = O\left(\frac{n^2}{p^{2/3}}\right) + O\left(\frac{n^2}{p^{\alpha(2-\log 3)}}\right)\\
&H(n, p) = 
\begin{cases}
O\left(\frac{n^2}{p^{2/3}}\right) & \alpha \geq \frac{2}{3(2-\log 3)} = \alpha_{\max}  \approx 1.59\\
O\left(\frac{n^2}{p^{\alpha(2-\log 3)}}\right) & \text{otherwise}
\end{cases}
\\
&S(n, p) = O(3^{k_0})
\end{align*}
Thus, we only need to consider $\alpha_{\min}\leq\alpha \leq \alpha_{\max}$, since $\alpha\geq\alpha_{\min}$ is needed for the optimal local computation cost, and for $\alpha>\alpha_{\max}$, the communication cost of matrix multiplication dominates, while the syncronization cost is higher than that of $\alpha = \alpha_{\max}$.
The parameter $\alpha$ describes a trade-off between communication and synchronization costs.
\section{Strip-recursive algorithm}
In this section, we describe an algorithm based on a different recursion scheme, which we call antidiagonal recursion. We partition the matrix into antidiagonal strips. The algorithm proceeds by splitting a strip into two half-sized strips and processing them recursively; between the two recursive calls, the matrix needs to be updated in a carefully controlled manner in order to optimize communication.

\subsection{Banded Matrices}
The algorithm operates on matrices of aim special form.
\begin{definition}
Matrix $A$ is called \emph{lower banded of bandwidth $k$} if
\[
A_{ij} \neq 0 \;\Rightarrow\; 0 \leq i - j \leq k.
\]
An \emph{upper banded} matrix is defined symmetrically.
\end{definition}

\begin{lemma}
Let matrices $A, B$ be lower banded of bandwidth $k$. Then the matrix $AB$ is lower banded of bandwidth $2k$.
\end{lemma}

\begin{lemma}
Let $A, B$ be $n \times n$ matrices, lower banded of bandwidth $m$. Then the product $AB$ can be computed in parallel on $p$ processors with the following costs:
\begin{gather*}
W  \textstyle = O\!\left(\frac{m^2 n}{p}\right),\qquad
H \textstyle = O\!\left(\frac{m^{4/3} n^{2/3}}{p^{2/3}}\right),\qquad
S \textstyle = O(1).
\end{gather*}
\end{lemma}
\begin{proof}
Each of the matrices contains $O(mn)$ nonzero elements. Each element of matrix $A$ is multiplied by $O(m)$ elements of matrix $B$, so the computation contains $O(m^2 n)$ elementary multiplications. Our goal is to partition these operations into equal near‑cubic cuboid blocks. We do this as follows: partition the nonzero part of matrix $A$ (i.e., the strip of width $m$ below main diagonal) into strips parallel to main diagonal of width $O(m^{2/3} n^{1/3} p^{1/3})$; then partition each strip horizontally into regular parallelograms of equal width (Figure~\ref{f-strip}). Proceed similarly with the second matrix, but partition the strips into parallelograms vertically.

\begin{figure}
\centering
\begin{tikzpicture}[scale=0.4]
%
    \draw[black] (0,0)--(0,8)--(8,0)--(0,0);
    \draw[black] (4,0)--(0,4);
    \draw[red] (6,0)--(0,6);
    \draw[red] (5,0)--(0,5);
    \draw[red] (7,0)--(0,7);
    \node[draw=none, fill=none] at (-1.7,7.5) {$\frac{m^{2/3}n^{1/3}}{p^{1/3}}$};
    \draw[decorate, decoration=brace] (0,7) -- (0,8);
    \draw[red] (0,7)--(1,7);
    \draw[red] (0,6)--(2,6);
    \draw[red] (0,5)--(3,5);
    \draw[red] (0,4)--(4,4);
    \draw[red] (1,3)--(5,3);
    \draw[red] (2,2)--(6,2);
    \draw[red] (3,1)--(7,1);
    \draw[dotted, red] (1,3)--(8,3);
    \draw[dotted, red] (1,4)--(8,4);
    \draw[decorate, decoration={brace, mirror}] (8,3) -- (8,4);
    \node[draw=none, fill=none] at (9.7,3.5) {$\frac{m^{2/3}n^{1/3}}{p^{1/3}}$};
  \end{tikzpicture}
\caption{\label{f-strip} Multiplication of lower banded matrices}
\end{figure}
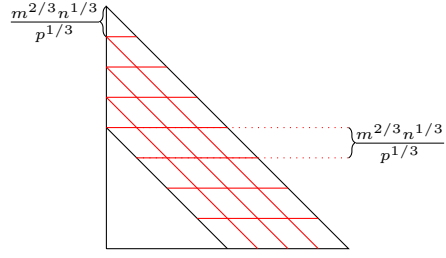

The product of each pair of blocks requires $O(m^2 n)$ operations. After computing the product of a pair of blocks, $O((m^{2/3} n^{1/3})^2) = O(m^{4/3} n^{2/3})$ values are obtained, which must be transferred between processors to compute the values of the product matrix. Thus, multiplying matrices of this type requires computation cost $W = O(m^2 n / p)$ and communication cost $H = O(m^{4/3} n^{2/3} / p^{2/3})$. The number of supersteps is obviously constant: $S = O(1)$.
\end{proof}
\subsection{Antidiagonal strips}

Next, we consider an algorithm based on the following recursive scheme. The matrix is partitioned into anti-diagonal strips of width $n_0 = n/p^{\alpha}$. The algorithm processes these strips sequentially, from left-to-right and top-down. In lemma 3 we describe parallel algorithm for reducing in BSP model. We call a strip \textit{reduced} if every row (respectively, column) of the matrix contains at most one nonzero entry within the strip, and that entry is $1$. A strip of width $n_0$ can be reduced by multiplying the matrix on the left and right by banded matrices of bandwidth $n_0$. 
Before reduction, preprocessing of the current strip is required in order to eliminate certain rows and columns. This elimination is performed by multiplying on the left and right by banded matrices of bandwidth $2n_0$.

\begin{lemma}
An antidiagonal strip of width $n_0$ of an $n \times n$ matrix can be reduced by multiplying the matrix on the left and right by banded matrices, which are also of size $n \times n$ and bandwidth $n_0$. This reduction can be performed in the BSP model with computation cost $W = O(n_0^3)$ and communication cost $H = O(n_0^2)$.
\end{lemma}

\begin{proof}
To reduce the strip, partition it horizontally into blocks of $n_0$ rows each (there will be between $1$ and $p$ blocks). Then each processor processes one block assigned to it. We perform the reduction using a sequence of transvections on each row and each column with rows below and columns to the right within the same block. Now each block contains at most one nonzero element in each row and each column. Each column intersects at most two blocks, hence contains at most two nonzero elements. Then each processor sends the resulting block to the processor handling the next block, and if necessary, row transvection is performed to eliminate the second nonzero element in the column. Note that each row is subtracted from at most $n_0$ rows below, and each column from at most $n_0$ columns to the right; therefore, the transvection matrices contain nonzero elements only in the strip of width $n_0$ below (above) the main diagonal. Thus, reduction of the strip is accomplished by multiplying on the left and right by lower‑banded and upper‑banded matrices of bandwidth $n_0$, respectively. Processing a block of the indicated size requires computation cost $O(n_0^3)$ and communication cost $O(n_0^2)$; the communication and computation costs of the second stage are $O(n_0^2)$ and $O(n_0)$, respectively.
\qed
\end{proof}

We introduce a numbering of the antidiagonals from left to right and top to bottom. Antidiagonal $r$ consists of elements $A_{i, j}$, where $i + j - 1 = r$. The strip consisting of consecutive antidiagonals numbered $m, m+1, \dots, m+k$ inclusive will be denoted by $A\langle m, m+k\rangle$.
\begin{lemma}
Suppose that strip $A\langle k, k + m\rangle$  has been reduced and let $i_1,\dots,i_k$, $j_1,\dots,j_k$ be the indices of rows and columns, containing nonzero elements in this strip. By multiplying $A$ on the left by a lower‑banded matrix $L$ and on the right by an upper‑banded matrix $U$, both of bandwidth $2m$, we can obtain a matrix $A'$, where $A'\langle k+m, k + 2m\rangle$ consists only of zero elements in rows $i_1,\dots,i_k$ and columns $j_1,\dots,j_k$.
\end{lemma}

\begin{proof}
Let us construct the elimination matrix for the columns (the elimination matrix for the rows is constructed similarly). For all $i, j$, where $A_{ij} = 1$ and $k+1\leq i + j \leq k+m+1$, i.e., the element $A_{ij}$ belongs to strip $A\langle k,k+m\rangle$, let $L_{r,i} = -A_{r,j}$ for $k+m +1 - j \leq r \leq k+ 2m +1 - j$. 
Other elements of $L$ are defined as $L_{i,j} = 
\begin{cases}
    0, & i\neq j \\
    1, & i=j
\end{cases}
$. Direct computation verifies that matrix $L$ does indeed eliminate the required columns. It is also clear that it is lower‑banded of bandwidth $2m$.
\qed
\end{proof}

\subsection{Recursive scheme}
The algorithm takes strip $A\langle k, k + m\rangle$ as input and reduces it as follows:

\paragraph{1. Reduction of $A\langle k, k +m/2\rangle$} This is performed recursively, thus obtaining elimination matrices $L_m$ and  $U_m$ of bandwidth $2m$; $L_mAU_m = A_1$
\paragraph{2. Update on $A_1\langle k + m/2 +1, m\rangle$} Thus obtaining elimination matrices $L'_m$ and $U'_m$, with which we eliminate certain rows and columns in the lower-right strip: more precisely, those that already contain a nonzero element in $A\langle k, k + m/2\rangle$. Matrices $L'_m$ and $U'_m$ are lower-banded and upper‑banded of bandwidth $2m$, respectively; $L'_mA_1U_m = A_2$
\paragraph{3. Reduction of $A_2\langle k+ m/2 +1, m\rangle$} This is performed recursively, thus obtaining elimination matrices $\tilde{L}_m$ and $\tilde{U}_m$ (again a lower‑banded and an upper‑banded matrix of bandwidth $2m$); $\tilde{L}_mA_2\tilde{U}_m = A_{Final}$
\paragraph{Summary} Overall, $A\langle k, k + m\rangle$ is reduced
by multiplying on the left by lower-triangular elimination matrices and on the
right by upper-triangular elimination matrices
$\tilde{L}_mL'_mL_mAU_mU'_m\tilde{U}_m = A_{Final}$, where $A_{Final}\langle k, k+m\rangle$ is reduced.

\subsection{Cost analysis}

A recursive call on a strip of bandwidth $m$ performs two recursive calls on
strips of bandwidth $m/2$ and several banded matrix products, as described in steps 2 and 4. These matrix products are performed in parallel similarly to the proof of Lemma 4 with the following costs:
\begin{gather*}
\textstyle 
W = O\!\left(\frac{m^2 n}{p}\right),\quad
H = O\!\left(\frac{m^{4/3} n^{2/3}}{p^{2/3}}\right),\quad
S = O(1).
\end{gather*}

Next, we need to compute $L_{2m} = L_m L'_m L_m$ and $U_{2m} = U_m U'_m U_m$; this is done with the same costs, according to Lemma 4.

Thus, we can derive the following recurrence relations:
\begin{align*}
W(m, p) &=\textstyle 2\,W(m/2, p) + O\!\left(\frac{m^2 n}{p}\right),\\
H(m, p) &=\textstyle 2\,H(m/2, p) + O\!\left(\frac{m^{4/3} n^{2/3}}{p^{2/3}}\right),\\
S(m, p) &=\textstyle 2\,S(m/2, p) + O(1).
\end{align*}
The base values for the recurrences are provided by the respective costs of reducing a strip of width $n_0$, as described in Lemma 3:
\begin{align*}
W(n_0, q) = O(n_0^3) \qquad
H(n_0, q) = O(n_0^2) \qquad
S(n_0, q) = O(1)
\end{align*}

Since the recursion depth is $\alpha\log p$, and $n_0 = n/p^{\alpha}$, the recurrence relations are solved as follows:
\begin{gather*}
W(n, p) =\textstyle \sum_{i=0}^{\alpha\log p} O\!\Bigl(\frac{(2^i n_0)^2 n}{p}\Bigr)
 +2^{\alpha\log p}O({n_0}^{3}) 
= O\Bigl(\frac{n^3}{p}\Bigr)+O\Bigl(\frac{n^3}{p^{2\alpha}}\Bigr) = O\Big(\frac{n^3}{p}\Bigr),
\\
H(n, p) =\textstyle \sum_{i=0}^{\alpha \log p} O\Bigl(\frac{(2^i n_0)^{ \frac{4}{3}} n^{ \frac{2}{3}}}{p^{\frac{2}{3}}}\Bigr)+2^{\alpha\log p}O({n_0}^{2}) \\
 = \sum_{i=0}^{\alpha \log p} 2^{\frac{4i}{3}}\, O\Bigl(\frac{n_0^{\frac{4}{3}} n^{\frac{2}{3}}}{p^{\frac{2}{3}}}\Bigr) + 2^{\alpha\log p}O({n_0}^{2})
=\textstyle p^{\frac{4\alpha}{3}}\, O\Bigl(\frac{n^2}{p^{\frac{2}{3}
+\frac{4\alpha}{3}}}\Bigr) + p^{\alpha}O\Bigl(\frac{n^2}{p^{2\alpha}}\Bigr)\\
=\textstyle O\Bigl(\frac{n^2}{p^{\frac{2}{3}}}\Bigr) + O(\frac{n^2}{p^{\alpha}}) =
\begin{cases}
O\Bigl(\frac{n^2}{p^{\frac{2}{3}}}\Bigr) & \alpha \geq \alpha_{\max} = \frac{2}{3}
\\
O\Bigl(\frac{n^2}{p^{\alpha}}\Bigr) &  \text{otherwise}
\end{cases}
\\
S(n, p) \textstyle = O(2^{\alpha \log p}) = O(p^{\alpha}).
\end{gather*}
Since we need $W(n, p) = O\Bigl(\frac{n^3}{p}\Bigr)$ it must hold $\alpha\geq \alpha_{\min} =  1/2$.
Similarly to the block-recursive algorithm, parameter $\alpha$ describes a trade-off between communication and synchronization costs.
Thus, we only need to consider $\alpha_{\min}\leq\alpha \leq \alpha_{\max}$, since $\alpha\geq\alpha_{\min}$ is needed for the optimal local computation cost, and for $\alpha>\alpha_{\max}=  \frac{2}{3}$, the communication cost of matrix multiplication dominates, while the synchronization cost is higher than that of $\alpha = \alpha_{\max}$.
\section{Conclusion}
We have developed two algorithms for Bruhat decomposition in the BSP model: the block-recursive algorithm and the strip-recursive algorithm, and compared their effectiveness. The block-recursive algorithm uses a more familiar partition into blocks and a simpler recursive scheme. The strip-recursive algorithm requires some technical preparation (introducing special types of matrices and adapting existing parallel algorithms to such matrices). On the positive side, this algorithm uses a simpler recursion scheme and has lower cost. Both algorithms rely on efficient BSP matrix multiplication, and exhibit a trade-off between communication and synchronization costs. In the strip-recursive algorithm this trade-off is more efficient than in the block-recursive algorithm, matching the analogous trade-off by Tiskin \cite{Tiskin:07_FGCS} for $LU$-decomposition and generic pairwise elimination.
\bibliographystyle{plain}
\bibliography{research}

\end{document}